\documentclass[preprint]{imsart}

\RequirePackage{amsthm,amsmath,amsfonts,amssymb,bm,bbm, float, subfig, enumerate, enumitem}
\RequirePackage[numbers]{natbib}
\RequirePackage[colorlinks,citecolor=blue,urlcolor=blue,backref=page,backref=page]{hyperref}
\RequirePackage{graphicx}

\pubyear{2026}
\arxiv{2501.10263}
\volume{TBA}
\issue{TBA}
\firstpage{1}
\lastpage{1}

\startlocaldefs
\theoremstyle{plain}

\newtheorem{theorem}{Theorem}[section]

\theoremstyle{definition}
\newtheorem{definition}[theorem]{Definition}
\newtheorem{example}{Example}

\theoremstyle{remark}

\def\RR{\mathbb R}
\newcommand{\E}{\operatorname{E}} 

\newtheorem{condition}{Condition}
\newtheorem{proposition}{Proposition}

\newcommand{\tr}{\operatorname{Tr}}

\endlocaldefs

\begin{document}

\begin{frontmatter}
\title{Prior distributions for structured semi-orthogonal matrices}
\runtitle{Priors for semi-orthogonal matrices}

\begin{aug}
\author[A]{\fnms{Michael}~\snm{Jauch}\ead[label=e1]{mjauch@fsu.edu}},
\author[B]{\fnms{Marie-Christine}~\snm{D{\"u}ker}\ead[label=e2]{marie.dueker@fau.de}}
\and
\author[C]{\fnms{Peter D.}~\snm{Hoff}\ead[label=e3]{peter.hoff@duke.edu}}
\address[A]{Department of Statistics,
Florida State University\printead[presep={,\ }]{e1}}
\address[B]{Department of Data Science,
FAU Erlangen-N{\"u}rnberg\printead[presep={,\ }]{e2}}
\address[C]{Department of Statistical Science,
Duke University\printead[presep={,\ }]{e3}}
\runauthor{M. Jauch et al.}
\end{aug}

\begin{abstract}
Statistical models for multivariate data often include a semi-orthogonal matrix parameter. In many applications, there is reason to expect that the semi-orthogonal matrix parameter satisfies a structural assumption such as sparsity or smoothness. From a Bayesian perspective, these structural assumptions should be incorporated into an analysis through the prior distribution. In this work, we introduce a general approach to constructing prior distributions for structured semi-orthogonal matrices that leads to tractable posterior inference via parameter-expanded Markov chain Monte Carlo. We draw on recent results from random matrix theory to establish a theoretical basis for the proposed approach. We then introduce specific prior distributions for incorporating sparsity or smoothness and illustrate their use through applications to biological and oceanographic data.
\end{abstract}


\end{frontmatter}

\section{Introduction}


Many statistical models for multivariate data include a parameter belonging to the set of semi-orthogonal matrices $\mathcal{V}(k,p) = \{ Q \in \RR^{p\times k} \mid  Q^\top Q =  I_k\} $ with $k < p$.  The set $\mathcal{V}(k,p)$ is referred to as the Stiefel manifold \citep{Chikuse2003}. In these models, the columns of the semi-orthogonal matrix parameter can often be understood as spanning a low-dimensional subspace in which much of the interesting variation in high-dimensional data occurs. The rows often represent the coordinates of statistical units embedded in a low-dimensional latent space. Examples appear throughout the literature on covariance estimation, principal component analysis (PCA), and factor models \citep{Cai2013, Gao2015, Franks2019, bryan2021multirank, Park2024, Xie2024}; low-rank matrix or tensor approximation \citep{Hoff2007, Hoff2016, choi2017selecting, gavish2017optimal,  Yuchi2023}; network data \citep{Hoff2009, Loyal2025, li2024bsnmani}; 
spatial statistics \citep{north2024flexible}; time series analysis \citep{meng2024bayesian}; and other topics. These models have been applied in diverse fields including climate science \citep{north2024flexible}, ecology \citep{Stouffer2021, Loyal2025}, economics \citep{Loyal2025}, neuroscience \citep{li2024bsnmani, meng2024bayesian, Park2024}, and seismic imaging \citep{Yuchi2023}.


\begin{example}[Model-based singular value decomposition] \label{ex:SVD} Suppose $Y$ is an $n \times p$ data matrix. It is common to assume that $Y = M + E$ where $M$ has rank $k \ll \min\{n, p\}$ and $E$ is a matrix of errors \citep{Hoff2007, choi2017selecting, gavish2017optimal, Yuchi2023}. 
The low-rank matrix $M$ can be represented in terms of its truncated singular value decomposition (SVD) $M = U D V^\top$ where $U\in \mathcal{V}(k,n),$ $V \in\mathcal{V}(k,p),$ and $D$ is a $k \times k$ diagonal matrix with positive entries on the diagonal. 
Let $y_i, u_i,$ and $e_i$ be the $i$th rows of $Y, U,$ and $E.$ The model-based SVD can be viewed as a factor model in which the $p$ measurements in $y_i$ are related to $k$ latent factors in $u_i$ through the equation $y_{i} = V D u_i+ e_{i}.$ The model-based SVD is also closely related to PCA. In fact, if $Y$ is centered and $e_{ij} \stackrel{\text{ind.}}{\sim} \text{Normal}(0, \sigma^2)$, the MLEs of $UD$ and $V$ coincide with the principal component scores and loadings from classical PCA. 
\end{example}

\begin{example}[Network eigenmodel] \label{ex:eigenmodel}
Now suppose $Y = (y_{ij})$ is a $p \times p$ adjacency matrix with $y_{ij} \in \{0,1\}$ indicating the presence or absence of a relationship between object $i$ and object $j.$ The network eigenmodel of \citet{Hoff2009} supposes that 
\begin{align*}y_{ij} &\mid \pi_{ij} \stackrel{\text{ind.}}{\sim} \text{Bernoulli}\left(\pi_{ij}\right) \\ 
\pi_{ij} &=  \Phi\left\{c + \left(Q \Lambda Q^\top\right)_{ij}\right\}
\end{align*}
where $\Phi$ is the cumulative distribution function of a standard normal random variable and $(c, Q, \Lambda)$ are unknown parameters. The parameter $Q$ is a $p \times k$ semi-orthogonal matrix, $\Lambda=\text{diag}(\lambda_1, \ldots, \lambda_k)$ is a $k \times k$ diagonal matrix, and $c$ is a real number. 
\end{example}


In applications, there is often reason to assume that a semi-orthogonal matrix parameter satisfies a structural assumption such as sparsity or smoothness. If each $y_i$ from Example~\ref{ex:SVD} contains measurements of some variable over a uniform, one-dimensional grid of points in time or space, then it may be reasonable to assume that the columns of $V$ will resemble discretizations of smooth curves. If the scientific context suggests that $y_{ij}$ should only relate to a subset of the latent factors in $u_i,$ then it may be reasonable to assume that $V$ is sparse. Such an assumption typically leads to more interpretable results. In some high-dimensional settings, making structural assumptions regarding a semi-orthogonal matrix parameter allows us to make meaningful inferences that would otherwise be impossible.

In Bayesian statistics, structural assumptions are incorporated into an analysis through the prior distribution. The vast majority of articles on Bayesian methodology for models with semi-orthogonal matrix parameters assign the semi-orthogonal matrix parameter a uniform prior. There are a handful of exceptions that propose non-uniform priors for structured semi-orthogonal matrices. \citet{Yoshida2010} introduced a prior for a sparse semi-orthogonal factor loadings matrix and a computational method combining variational inference with simulated annealing to find the posterior mode. \citet{Gao2015} proposed a prior for a sparse semi-orthogonal matrix and analyzed its theoretical properties in the context of sparse PCA. The authors described a method for computing the posterior mean, but the discrete parameter space rendered posterior simulation intractable. \citet{Pourzanjani2021} combined the Givens angle representation of a semi-orthogonal matrix \citep{Shepard2015} and the horseshoe prior \citep{carvalho2010horseshoe} to construct a prior for sparse PCA. \citet{Jauch2021a} leveraged the matrix angular central Gaussian (MACG) distribution \citep{Chikuse1990} to define a prior for a semi-orthogonal matrix parameter with smooth structure and conducted posterior simulation via parameter-expanded Markov chain Monte Carlo (MCMC). \citet{north2024flexible} developed a prior for a semi-orthogonal matrix parameter with column-specific smooth structure and demonstrated its value in applications to spatial data. In principle, the matrix Bingham–von Mises–Fisher family of distributions \citep{Khatri1977, Hoff2009}, which is well established in the directional statistics literature, could provide non-uniform priors for structured semi-orthogonal matrices. In practice, the normalizing constants can be prohibitively difficult to evaluate, limiting their use in hierarchical models. The limited extent of the literature discussing priors for structured semi-orthogonal matrices is due, in part, to the difficulty of defining probability distributions on the Stiefel manifold which reflect such prior information and lead to tractable posterior inference.  


This article introduces a family of prior distributions for structured semi-orthogonal matrices that leads to tractable posterior inference via the parameter-expanded MCMC methodology of \citet{Jauch2021a}. The family of priors, which includes the uniform distribution and a class of MACG distributions as particular cases, is defined by projecting an unconstrained real-valued random matrix onto the Stiefel manifold. We draw upon recent results from random matrix theory to establish that aspects of the distribution of the real-valued matrix are inherited by the distribution of the semi-orthogonal matrix. Thus, if we want the prior distribution for a semi-orthogonal matrix parameter to reflect structural assumptions, we can build those assumptions into the distribution of the real-valued matrix. We introduce specific prior distributions for incorporating sparsity or smoothness and illustrate their use through applications to biological and oceanographic data. 


We now outline the remainder of the article. Section 2 reviews important preliminaries, including the polar decomposition, the MACG distribution, and the parameter-expanded MCMC methodology of \citet{Jauch2021a}. Section 3 introduces the proposed family of prior distributions and presents theoretical results demonstrating that aspects of the distribution of the real-valued matrix are inherited by the distribution of the semi-orthogonal matrix. Section 4 discusses the details of incorporating sparse or smooth structure with the proposed family of prior distributions. As part of this discussion, we describe a novel continuous shrinkage prior for the latent real-valued random matrix introduced in the construction of the prior for the semi-orthogonal matrix. Section 5 considers applications to biological and oceanographic data. We conclude in Section 6 with a discussion. 

\section{Preliminaries} \label{sec:preliminaries} 

The proposed prior distribution is related to the MACG distribution of \citet{Chikuse1990} and leads to tractable posterior simulation via polar expansion, the parameter-expanded MCMC method of \citet{Jauch2021a}. Both the MACG distribution and polar expansion are built upon the polar decomposition. 
In this section, we briefly review each of these topics. The discussion of these preliminaries draws upon \citet{Jauch2021a}. 


\paragraph{The polar decomposition.} The polar decomposition is the unique representation of a full rank matrix $X \in \mathbbm{R}^{p\times k}$ as the product $X = Q_{X} S_X^{1/2}$ where $Q_{X} \in \mathcal{V}(k,p),$ $S_X$ is a $k \times k$ symmetric positive definite (SPD) matrix, and $S_X^{1/2}$ is the symmetric square root of $S_X.$ The components of the polar decomposition can be computed from $X$ as $Q_{X} = X(X^\top X)^{-1/2}$ and $S_{X}=X^\top X.$ In terms of the singular value decomposition $X=U D V^\top,$ we have $Q_{X}=U V^\top$ and $S_X^{1/2}=VDV^\top.$ The semi-orthogonal component $Q_{X}$ can be understood as the projection of $X$ onto $\mathcal{V}(k,p).$ More precisely, $Q_{X}$ is the closest matrix in $\mathcal{V}(k,p)$ to $X$ in the Frobenius norm, i.e. $Q_{X} = \text{argmin}_{Q\in \mathcal{V}(k,p)} \|X - Q\|_{F}$ \citep{higham1990fast}. 

\paragraph{The MACG distribution.} The random semi-orthogonal matrix $Q \in \mathcal{V}(k,p)$ is said to have a matrix angular central Gaussian $\text{MACG}( \Sigma)$ distribution if $Q \stackrel{d}{=} Q_{X}$ where $X\sim N_{p,k}(0,\Sigma,I_k)$ \citep{Chikuse2003}. The notation $N_{p,k}( 0, \Sigma, I_k)$ indicates a mean-zero matrix normal distribution  with $\Sigma$ as its $p \times p$ row covariance matrix and the identity as its $k \times k$ column covariance matrix \citep{Srivastava1979, Dawid1981}. We can think of the MACG distribution as resulting from the projection of a mean-zero normal random matrix onto the Stiefel manifold. The $\text{MACG}(\Sigma)$ distribution has density $f (Q \mid \Sigma) =| \Sigma |^{-k/2}| Q^\top \Sigma^{-1} Q |^{-p/2}$ and is uniform on the Stiefel manifold when $\Sigma = I_p$ or $k=p.$


\paragraph{Posterior inference via polar expansion.} 
Among other problems, \citet{Jauch2021a} considered MCMC simulation from the Stiefel manifold when the target distribution is a posterior distribution arising from a $\text{MACG}(\Sigma)$ prior. Suppose we have data $y$ whose distribution given the unknown parameter $Q \in \mathcal{V}(k,p)$ has density $p(y \, \vert \, Q).$ The prior density is $p(Q)= |\Sigma|^{-k/2}|Q^\top\Sigma^{-1}Q|^{-p/2}.$ 
The posterior density satisfies $p(Q\, \vert \, y) \propto p(y \, \vert \, Q) \, p(Q).$ 

To avoid the complications of direct simulation from the posterior density $p(Q\, \vert \, y)$ defined on the Stiefel manifold, we can expand the parameter space by working with the latent real-valued random matrix $X$ implied by the definition of the MACG distribution. The prior distribution for $X$ is the mean-zero matrix normal distribution $N_{p,k}(0,\Sigma, I).$ The posterior density satisfies 
\begin{align}
    p(X \mid y) &\propto p(y \mid Q_{X}) \, N_{p,k}(X \mid 0, \Sigma, I_k). \label{MACGposterior}
\end{align} 
To obtain an MCMC approximation to the posterior distribution of $Q,$ we can simulate a Markov chain $X_1,\dots, X_T$ on $\mathbbm{R}^{p\times k}$ whose stationary distribution has density \eqref{MACGposterior}, calculate $Q_t = X_t(X_t^\top X_t)^{-1/2}$ for each iteration $t,$ and take the empirical distribution of $Q_1,\dots, Q_T$ as our approximation. 
The posterior distribution of $X$ is non-standard, but knowing its density (up to a constant factor) allows us to apply standard MCMC methods, e.g. adaptive Hamiltonian Monte Carlo (HMC) \citep{Neal2011, Hoffman2014} as implemented in Stan \citep{Carpenter2017} . 

An important idea, not pursued or discussed in \citet{Jauch2021a}, is that the same approach to posterior inference is available for other priors defined by projecting an unconstrained real-valued random matrix $X$ onto the Stiefel manifold, including the family of prior distributions introduced in this article. We simply replace the mean-zero matrix normal density in \eqref{MACGposterior} with an alternative prior density for $X.$ This computational tractability 
is an important advantage of the proposed family of prior distributions. 

\section{The family of prior distributions and its properties} \label{sec:the_family_and_properties}


We now introduce the proposed family of prior distributions. Let $Z=(Z_{ij})$ be a $p \times k$ matrix with i.i.d. real entries having mean zero and unit variance; let $\Omega$ be a $p \times p$ correlation matrix; and set $  X = \Omega^{1/2}  Z.$ The proposed prior is the distribution of $ Q_X,$ the projection of $X$ onto the Stiefel manifold $\mathcal{V}(k,p)$ obtained via the polar decomposition. There are two notable special cases.  If the entries of $Z$ are independent standard normal random variables, then $Q_X \sim \text{MACG}(\Omega).$ If, additionally, we have that $\Omega = I_p,$ then $Q_X$ is uniform on $\mathcal{V}(k,p).$

In the remainder of this section, we present results demonstrating that properties of the distribution of $X$ are inherited by the distribution of $Q_X.$ Theorem~\ref{prop:invariance} establishes that invariance properties of $X$ are shared by $Q_X.$ Theorem~\ref{thm:limitingdistribution} implies that the Wasserstein distance between the distribution of $m$ entries of $\sqrt{p}Q_X$ and the distribution of the corresponding entries of $X$ converges to zero as $k\to \infty$ and $p = p(k) \to\infty$ with $k/p \to 0,$ provided that $m = m(k) = o\left(p/k\right).$ 
These results have important implications for statistical modeling. If we want the prior distribution for a semi-orthogonal matrix parameter to reflect structural assumptions, we can build these features into the distribution of $X.$ Theorem~\ref{thm:limitingdistribution} is particularly relevant to prior specification in the common scenario in which a semi-orthogonal matrix parameter has far more rows than columns. Beyond its relevance to statistical modeling, Theorem~\ref{thm:limitingdistribution} may be of independent interest to those working in high-dimensional probability and random matrix theory. 

\subsection{Invariance properties} \label{se:invariance}


Let $\mathcal{L} \subseteq \mathcal{O}(p)$ and $\mathcal{R} \subseteq \mathcal{O}(k)$ where $\mathcal{O}(p)$ and $\mathcal{O}(k)$ are the orthogonal groups in dimensions $p$ and $k.$ A $p \times k$ random matrix $Y$ is invariant to left multiplication by elements of $\mathcal{L}$ if $L Y \stackrel{d}{=} Y$ for all $L \in \mathcal{L}$. The random matrix $Y$ is invariant to right multiplication by elements of $\mathcal{R}$ if $Y R \stackrel{d}{=} Y$ for all $R \in \mathcal{R}$. Finally, the random matrix $Y$ is invariant to left multiplication by elements of $\mathcal{L}$ and right multiplication by elements of $\mathcal{R}$ if $L Y R \stackrel{d}{=} Y$ for all $L \in \mathcal{L}$ and  $R \in \mathcal{R}$. 

\begin{theorem}[Invariance] \label{prop:invariance}
If the random matrix $X$ is invariant to left multiplication by elements of $\mathcal{L}$ and right multiplication by elements of $\mathcal{R},$ then so is $Q_{X}$.
\end{theorem} 

Theorem~\ref{prop:invariance} characterizes how invariance properties of the distribution of $X$ influence those of the distribution of $Q_X$ and thus has important implications for choosing an appropriate prior from the proposed family. The uniform distribution on $\mathcal{V}(k,p)$ is the unique distribution which is invariant to left multiplication by elements of $\mathcal{O}(p)$ and right multiplication by elements of $\mathcal{O}(k).$ Theorem~\ref{prop:invariance} implies that our prior distribution is uniform when $X$ is invariant to left multiplication by elements of $\mathcal{O}(p)$ and right multiplication by elements of $\mathcal{O}(k).$ When we expect a semi-orthogonal matrix parameter to have structure such as sparsity or smoothness, a uniform prior is inappropriate, but less stringent invariance properties may still be desirable. For example, if row indices provide no substantive information, our prior distribution should be invariant to left multiplication by permutation matrices. If, additionally, there is no prior information regarding the signs of its entries, our prior distribution should be invariant to left multiplication by signed permutation matrices. Theorem~\ref{prop:invariance} tells us how to achieve this. 

\subsection{Wasserstein approximation} \label{se:wass_approx}


Theorem~\ref{thm:limitingdistribution} considers the Wasserstein distance between the distribution of $m$ entries of $\sqrt{p}Q_X$ and the distribution of the corresponding entries of $X$ as $k$ and $p$ grow. An important implication is that the Wasserstein distance goes to zero as $k \to \infty$ and $p=p(k)\to \infty$ with $k/p \to 0,$ provided that $m=m(k) = o\left(p/k\right).$ 

\begin{definition} \label{def:Wasserstein}
Let $(\mathcal{X}, d)$ be a Polish metric space. The Wasserstein distance of order $\ell \in [1, \infty)$ between two probability measures $\mu, \nu$ on $\mathcal{X}$ with finite $\ell$th moments is defined as
\begin{equation}
W_{\ell}(\mu, \nu) = \left\{ \inf_{\gamma \in \Gamma(\mu, \nu)} \int_{\mathcal{X} \times \mathcal{X}} d(x,y)^{\ell} \, d\gamma(x, y) \right\}^{1/\ell}
\end{equation}
where \(\Gamma(\mu, \nu)\) is the set of all joint distributions (or couplings) \(\gamma\) with marginals \(\mu\) and \(\nu\).
\end{definition}
\noindent The Wasserstein distance metrizes weak convergence \cite[Theorem 6.9]{Villani2009}. In our results, we consider the Wasserstein distance of order $\ell=2$ with respect to 
the Euclidean norm. 


To emphasize that we are considering an asymptotic setting in which the dimensions of the random matrices and, potentially, the number of elements selected grow, we will include subscripts. For example, we will denote $X$ as $X_k,$ $Q_{X}$ as $Q_{X_k},$ and $\Omega$ as $\Omega_p.$ It will be implicit that $p$ and $m$ grow as functions of $k.$ 


Recall that $X_k = \Omega_{p}^{1/2} Z_k.$ We impose the following conditions on the matrices  $Z_k$ and $\Omega_{p}:$
\begin{condition} \label{ass:moments}
The entries of $Z_k \in \RR^{p \times k}$ are i.i.d. with mean zero, unit variance,  and $\E |Z_{ij}|^{4 + \delta} < \infty$ for some $\delta > 0$.
\end{condition}
%
\begin{condition} \label{ass:Omega}
 The matrix $\Omega_{p} \in \RR^{p \times p}$ is a correlation matrix for which, as $p \to \infty$, 
 \begin{itemize}
\item the quantity $c_{\Omega}(p) := p^{-1}\tr\left(\Omega_{p}^2\right)$ converges to a constant $\bm{c}_{\Omega} >0,$ and 
\item the spectral norm $\|\Omega_p\|$ is bounded above.
\end{itemize}
\end{condition}
\noindent Conditions similar to these are common in the random matrix theory literature. Proposition~\ref{prop:Basic_facts} in Section~\ref{se:Incorporating_Structure} verifies that Condition~\ref{ass:moments} holds for the novel shrinkage prior proposed in Section~\ref{se:Incorporating_Structure}. Proposition~\ref{prop:corr} in Section~\ref{se:Incorporating_Structure} verifies that Condition~\ref{ass:Omega} holds for a variety of different sequences of correlation matrices.

Before stating Theorem~\ref{thm:limitingdistribution}, we must introduce additional notation related to the selection of $m$ entries from the matrices $\sqrt{p}Q_{X_k}$ and $X_k.$  For each $j \in \{1, \ldots, k\},$ let $\mathcal{M}^{(j)}_k \subseteq \{1, \ldots, p\}$ contain the row indices of the entries to be selected from column $j$ of $\sqrt{p}Q_{X_k}$ and $X_k.$ Setting $m_k^{(j)} = \left|\mathcal{M}_k^{(j)}\right|,$ we have that $m = \sum_{j=1}^{k} m_k^{(j)}.$ For each $j \in \{1, \ldots, k\}$ such that $\mathcal{M}^{(j)}_k \neq \emptyset,$ define $D_{\mathcal{M}_k^{(j)}}$ as the $m_k^{(j)} \times p$ matrix obtained from the identity matrix $I_p$ by removing row $i$ if $i \notin \mathcal{M}_k^{(j)}.$ Let $e_{j}$ denote the $j$th standard unit vector in $\RR^{k},$ and define the matrix
\begin{align*}
\Delta_{\mathcal{M}_k^{(1)}, \ldots, \mathcal{M}_k^{(k)}} = 
\begin{bmatrix}
e_1^\top \otimes D_{\mathcal{M}_k^{(1)}} \\ 
e_2^\top \otimes D_{\mathcal{M}_k^{(2)}} \\ 
\ldots \\ 
e_k^\top \otimes D_{\mathcal{M}_k^{(k)}}
\end{bmatrix}.
\end{align*}
(If $\mathcal{M}^{(j)}_k = \emptyset$ so that $D_{\mathcal{M}_k^{(j)}}$ is undefined, we omit the block $e_j^\top \otimes D_{\mathcal{M}_k^{(j)}}$ when constructing $\Delta_{\mathcal{M}_k^{(1)}, \ldots, \mathcal{M}_k^{(k)}}.$) The $m$-dimensional vectors 
\begin{align*}
\Delta_{\mathcal{M}_k^{(1)}, \ldots, \mathcal{M}_k^{(k)}}  \operatorname{vec}\left(X_k\right) \quad  \text{and} \quad  \Delta_{\mathcal{M}_k^{(1)}, \ldots, \mathcal{M}_k^{(k)}} \operatorname{vec}\left(\sqrt{p}Q_{X_k}\right)
\end{align*}
contain the entries of $X_k$ and $\sqrt{p}Q_{X_k}$ specified by the sets $\mathcal{M}_k^{(1)}, \ldots, \mathcal{M}_k^{(k)}.$ 
\begin{theorem}[Wasserstein distance] \label{thm:limitingdistribution}
Suppose that $X_k = \Omega_p^{1/2} Z_k$ with $Z_k$ as in Condition~\ref{ass:moments} and $\Omega_p$ as in Condition~\ref{ass:Omega}. Then 
\begin{align} \label{Theorem2Distance}
W_2\left[\operatorname{law}\left\{\Delta_{\mathcal{M}_k^{(1)}, \ldots, \mathcal{M}_k^{(k)}} \operatorname{vec}\left(\sqrt{p}Q_{X_k}\right)\right\}, \operatorname{law}\left\{\Delta_{\mathcal{M}_k^{(1)}, \ldots, \mathcal{M}_k^{(k)}} \operatorname{vec}\left(X_k\right)\right\}\right] 
=
O\left( 
\sqrt{\frac{m k}{p}} \right)
\end{align}
as $k\to \infty$ and $p=p(k)\to\infty$ with $k/p \to 0.$
\end{theorem}
\noindent An immediate consequence of \eqref{Theorem2Distance} is that the Wasserstein distance goes to zero as $k \to \infty$ and $p=p(k) \to\infty$ with $k/p\to 0,$ provided that $m=m(k) = o\left(p/k\right).$ To make the implications of Theorem~\ref{thm:limitingdistribution} more concrete, we will consider specific asymptotic scenarios in Example~\ref{ex:asymptotic_sparse} and Example~\ref{ex:asymptotic_smooth} of Section~\ref{se:Incorporating_Structure}, after introducing prior distributions for semi-orthogonal matrix parameters with sparse or smooth structure.

There is significant interest in normal approximations to the entries of random orthogonal or semi-orthogonal matrices in the random matrix theory community. For a detailed review, see Section 6 of \citet{Jauch2020} and the references therein. We highlight an especially close connection between Theorem~\ref{thm:limitingdistribution} and this literature. A result of \citet{Stam1982} showed that, after rescaling by $\sqrt{p},$ the distribution of $m$ entries of a uniformly distributed unit vector 
can be approximated by the distribution of $m$ independent standard normal random variables as the length $p$ of the unit vector grows. \citet{Watson1983} extended this result to a uniformly distributed semi-orthogonal matrix with a growing number of rows but a fixed number of columns. Theorem~\ref{thm:limitingdistribution} offers a new perspective on this topic. 
It shows that the results of \citet{Stam1982} and \citet{Watson1983} describe special cases of a more general phenomenon. That is, the distribution of $m$ entries of the projection $Q_{X_k}$ of a real random matrix $X_k$ onto the Stiefel manifold $\mathcal{V}(k,p)$ can, after rescaling by $\sqrt{p},$ be approximated asymptotically by the distribution of the corresponding entries of $X_k,$ not only when the entries of $X_k$ are i.i.d. standard normal random variables, but whenever $X_k$ satisfies Condition \ref{ass:moments} and Condition \ref{ass:Omega}. 

Proposition 1 in the supplementary material, while not directly relevant to constructing prior distributions for structured semi-orthogonal matrices, may also be of interest to some readers. It presents a simple, non-asymptotic expression for the squared Wasserstein distance between the distributions of $\sqrt{p}Q_X$ and $X.$ This distance can be written in terms of the eigenvalues of a sample covariance matrix that has been the subject of several recent articles in the random matrix theory literature. 

\section{Incorporating Structure}
\label{se:Incorporating_Structure}

We now discuss specific prior distributions from the proposed family that are appropriate for semi-orthogonal matrix parameters with sparse or smooth structure. Section~\ref{se:sparsity} introduces a novel continuous shrinkage prior for the entries of $Z$ that leads to an approximately sparse $Q_X$ when $k \ll p.$ Section~\ref{se:smoothness} considers different correlation matrices $\Omega$ for incorporating row dependence, including some that lead to a $Q_X$ with smooth structure, and verifies that Condition~\ref{ass:Omega} holds for sequences of these correlation matrices.

\subsection{Sparsity}
\label{se:sparsity}

To construct a prior distribution for a sparse semi-orthogonal matrix, we will set $\Omega = I_p$ and let the entries of $Z$ be i.i.d. from a continuous shrinkage prior with mean zero, unit variance, and finite fourth moment. Then, by Theorem~\ref{thm:limitingdistribution}, $Q_{X}$ will inherit the approximate sparsity of $X$ in the high-dimensional setting where $k \ll p.$ 

Condition~\ref{ass:moments} rules out many choices for the continuous shrinkage prior. For example, global-local shrinkage priors 
such as the horseshoe \citep{carvalho2010horseshoe}, generalized double Pareto \citep{armagan2013generalized}, or Dirichlet-Laplace \citep{Bhattacharya2015}
introduce dependence through the global scale parameter and, in many cases, have infinite variance and fourth moments. These observations motivate us to propose a novel continuous shrinkage prior for the entries of $Z$ that satisfies the requirements above and has a number of other appealing properties. In particular, we propose to let 
the entries of $Z$ be i.i.d. real-valued random variables with common density
\begin{align} \label{zdensity}
    f(z \mid \ell) = \frac{(\ell/2)^{\ell/2}}{\Gamma(\ell/2)}|z|^{\ell-1}\exp{(-\ell z^2/2)} 
\end{align}
where $z\in \mathbb{R}$ and $\ell \in (0,1].$

Proposition~\ref{scalemixture} provides a representation of 
this distribution as a scale mixture of normal distributions, which is helpful for establishing its theoretical properties and leads to improved computational efficiency when we consider posterior inference in Section~\ref{sparse_network_eigenmodel}.
\begin{proposition}
\label{scalemixture}
For $\ell \in (0,1),$ let each entry of $Z$ be generated as $$Z_{ij} \, \vert\,  \theta_{ij} \stackrel{\text{ind.}}{\sim} \text{Normal}\left(0, \theta_{ij}/\ell\right)$$ where $\theta_{ij} \stackrel{\text{i.i.d.}}{\sim} \text{Beta}\left[\ell/2, \left(1-\ell\right)/2\right].$ The entries of $Z$ are then i.i.d. with density \eqref{zdensity}.
\end{proposition}

Proposition~\ref{prop:Basic_facts} presents some other important properties of the distribution. 
\begin{proposition}
\label{prop:Basic_facts}
The following properties hold: 
\begin{enumerate}[label=\Roman*.,ref=Part \Roman*.]
\item When $\ell = 1,$ $f(z \mid \ell)$ is the standard normal density. 
As $\ell \to 0,$ the kernel $z \mapsto |z|^{\ell-1}\exp{(-\ell z^2/2)}$ of the density \eqref{zdensity} converges pointwise to the function $z \mapsto \left|z\right|^{-1}.$
\label{item1.P3}
\item For $\ell \in (0,1),$ $\lim_{z \to 0} f(z \mid \ell) = \infty.$
\label{item2.P3}
\item The entries of the random matrix $Z$ are i.i.d. with mean zero, unit variance, and finite fourth moment. 
\label{item3.P3}
\item The squared entries of $Z$ are distributed $Z_{ij}^2 \stackrel{\text{i.i.d.}}{\sim} \text{Gamma}(\text{shape} =\ell/2, \,\text{scale} = 2/\ell).$
\label{item4.P3}
\end{enumerate}
\end{proposition}
\noindent Property I tells us that the proposed distribution includes the standard normal distribution (when $\ell = 1$) and the improper normal-Jeffreys prior (as $\ell \to 0$) as limiting cases. As discussed earlier, when the entries of $Z$ are i.i.d. standard normal random variables, $Q_{X}$ is uniform on $\mathcal{V}(k,p).$ The normal-Jeffreys prior was studied in \citet{Figueiredo2003} and \citet{Bae2004} as a means of obtaining sparse estimates in regression problems. Property II indicates that, like the univariate horseshoe density, the density \eqref{zdensity} is unbounded at zero. Property III confirms that the proposed continuous shrinkage prior satisfies the requirements we stated at the beginning of this subsection. Proposition IV establishes that the squared entries of $Z$ are independent gamma random variables with a common rate parameter, which allows us to derive the marginal distribution of the entries of $Q_{X}$ when $k=1$ in Example~\ref{ex:marginal}.

\begin{example} \label{ex:marginal} Suppose that $Z$ and $Q_{X}$ each have a single column so that $Z=(z_1, \dots, z_p)^\top$ and $Q_{X}=(q_1, \dots, q_p)^\top.$ In that case, the squared entries of $Q_{X}$ have joint distribution $\left(q_1^2, \dots, q_p^2\right)^\top \sim \text{Dirichlet}(\ell/2, \dots, \ell/2)$ and marginal distribution $q_i^2 \sim \text{Beta}\left(\alpha, \beta\right)$ where $\alpha = \ell/2$ and $\beta = (p-1)\ell/2.$ The entries of $Q_{X}$ themselves have marginal density
\begin{align*}
 q \mapsto \frac{1}{B(\alpha, \beta)}|q|^{2\alpha-1}(1 - q^2)^{\beta-1}, \quad -1 < q < 1.
\end{align*} 
Figure~\ref{fig:marginalcomparison} compares the densities of $z_i$ with those of $\sqrt{p}q_i$ for $\ell=.1$ when $p=5$ and $p=100.$ In the high-dimensional case when $p=100,$ these densities are indistinguishable, as we might expect based on Theorem~\ref{thm:limitingdistribution}. 
\end{example}

\begin{example} \label{ex:sparse_figure}
Figure~\ref{fig:sparse} compares the columns of a single realization $X^*$ of $X$ with those of $\sqrt{p}\,Q_{X^*}$ for the sparsity-inducing prior with $\ell = .1$ and dimensions $p=100, k=3.$ The entries of $X^*$ appear as gray dots while the entries of $\sqrt{p}\,Q_{X^*}$ appear as black circles. 
\end{example}

\begin{example} \label{ex:asymptotic_sparse}
We consider the proposed prior for a sparse semi-orthogonal matrix in the asymptotic setting of Theorem~\ref{thm:limitingdistribution}. For each $k,$ let $Z_k$ be a $p \times k$ matrix whose real-valued entries are i.i.d. with density~\eqref{zdensity}, and let $\Omega_p = I_p.$ Then set $X_k = \Omega_p^{1/2}Z_k$ and $Q_{X_{k}} = X_k\left(X_k^\top X_k\right)^{-1/2}.$ For concreteness, we will examine the asymptotic behavior of three entries of $\sqrt{p}Q_{X_k}$ and $X_k:$ those with row and column indices $(1,1), (3,1)$, and $(2,2).$ We can select these entries from $\sqrt{p}Q_{X_k}$ and $X_k$ using the matrix $\Delta_{\mathcal{M}_k^{(1)}, \ldots, \mathcal{M}_k^{(k)}}$ introduced in Section~\ref{se:wass_approx} by letting $\mathcal{M}_k^{(1)}=\{1,3\}$ and $\mathcal{M}_k^{(2)}=\{2\}$ with $\mathcal{M}_k^{(j)}=\emptyset$ for $j = 3, \ldots, k.$ The vectors $\Delta_{\mathcal{M}_k^{(1)}, \ldots, \mathcal{M}_k^{(k)}} \operatorname{vec}\left(\sqrt{p}Q_{X_{k}}\right)$ and $\Delta_{\mathcal{M}_k^{(1)}, \ldots, \mathcal{M}_k^{(k)}} \operatorname{vec}\left(X_k\right)$ contain the desired entries. By construction, $\Delta_{\mathcal{M}_k^{(1)}, \ldots, \mathcal{M}_k^{(k)}} \operatorname{vec}\left(X_k\right) \stackrel{d}{=} \varphi$ where $\varphi$ is a random vector containing three i.i.d. entries, each with density \eqref{zdensity}. Proposition~\ref{prop:Basic_facts} establishes that Condition~\ref{ass:moments} of Theorem~\ref{thm:limitingdistribution} is met, and it is easy to verify that Condition~\ref{ass:Omega} of Theorem~\ref{thm:limitingdistribution} is met as well. Thus, Theorem~\ref{thm:limitingdistribution} implies that
\begin{align} \label{wassconvergence_sparse}
W_2\left[\operatorname{law}\left\{\Delta_{\mathcal{M}_k^{(1)}, \ldots, \mathcal{M}_k^{(k)}} \operatorname{vec}\left(\sqrt{p}Q_{X_k}\right)\right\}, \operatorname{law}\left\{\varphi\right\}\right] \to 0
\end{align}
as $k \to \infty$ and $p=p(k) \to \infty$ provided that $k/p \to 0.$ Because the Wasserstein distance metrizes weak convergence, we can conclude from \eqref{wassconvergence_sparse} that
\begin{align*}
\Delta_{\mathcal{M}_k^{(1)}, \ldots, \mathcal{M}_k^{(k)}} \operatorname{vec}\left(\sqrt{p}Q_{X_k}\right) \stackrel{d}{\to} \varphi
\end{align*}
as $k \to \infty$ and $p=p(k) \to \infty$ provided that $k/p \to 0.$
\end{example}

\begin{figure}[H]  
\centerline{\includegraphics[width=\linewidth]{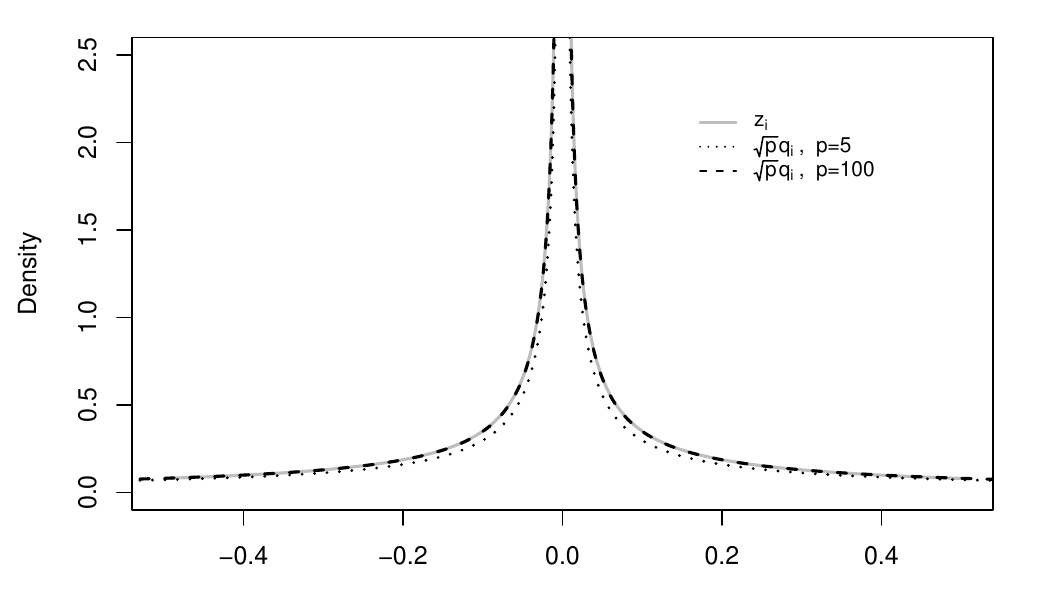}}
\caption{A comparison of the marginal densities of $z_i$ and $\sqrt{p}q_i$ for $\ell=.1$ when $p=5$ and $p=100,$ as discussed in Example~\ref{ex:marginal}.}
\label{fig:marginalcomparison}
\end{figure}

\begin{figure}[H]
\centerline{\includegraphics[width=\linewidth]{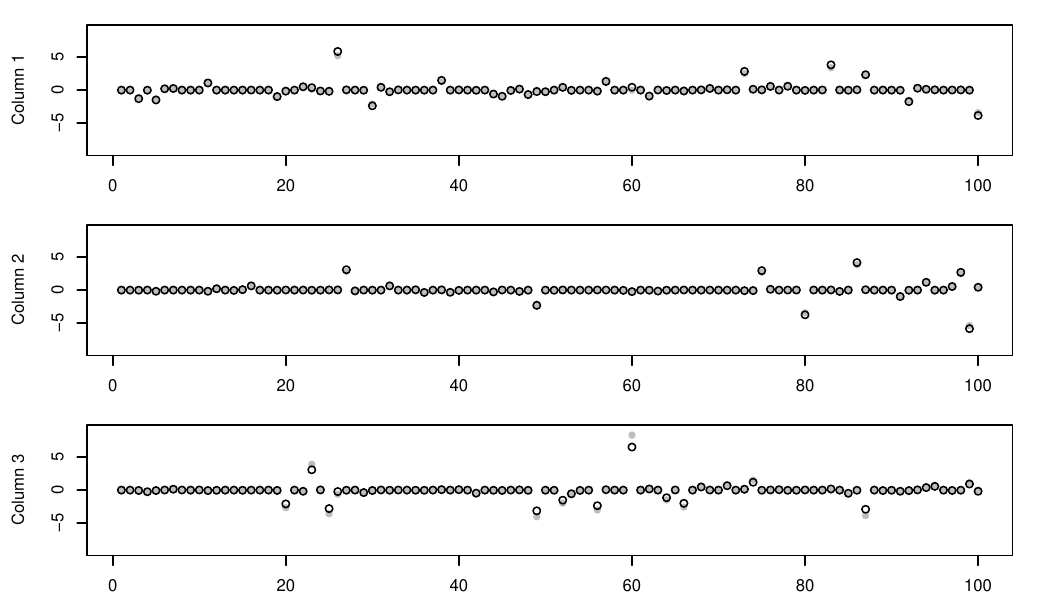}}
\caption{A comparison of the columns of a single realization $X^*$ of $X$ with those of $\sqrt{p}\,Q_{X^*}$ for the sparsity-inducing prior with $\ell = .1$ and dimensions $p=100, k=3$ as described in Example~\ref{ex:sparse_figure}. The entries of $X^*$ appear as gray dots while the entries of $\sqrt{p}\,Q_{X^*}$ appear as black circles.}
\label{fig:sparse}
\end{figure}

\subsection{Smoothness}
\label{se:smoothness}

To construct a prior distribution for a semi-orthogonal matrix with smooth structure, we introduce row dependence through the correlation matrix $\Omega.$
Example~\ref{ex:matern} considers the case where the entries of $Z$ are i.i.d. standard normals and $\Omega$ is defined with the Mat\'ern correlation function. Proposition~\ref{prop:corr} verifies that Condition~\ref{ass:Omega} of Theorem~\ref{thm:limitingdistribution} holds for a variety of different sequences of correlation matrices. Example~\ref{ex:asymptotic_smooth} examines the setup of Example~\ref{ex:matern} in the asymptotic setting of
Theorem~\ref{thm:limitingdistribution}.

\begin{example}  \label{ex:matern} Let the entries of $Z$ be i.i.d. standard normal random variables and let $\Omega = (\Omega_{ij})$ satisfy $\Omega_{ij} =  C^{\text{Mat\'ern }}_{\rho,\nu}\left(\left|i - j \right|\right)$ where
\begin{equation} \label{materncorrfunction}
    C^{\text{Mat\'ern }}_{\rho,\nu}(d) = \frac{2^{1-\nu}}{\Gamma(\nu)}\Bigg(\sqrt{2\nu}\frac{d}{\rho}\Bigg)^\nu K_\nu\Bigg(\sqrt{2\nu}\frac{d}{\rho}\Bigg)
\end{equation}
is the Mat\'ern correlation function \citep{Rasmussen2006} with parameters $\rho, \nu > 0.$ In this case, $Q_X \sim \text{MACG}( \Omega).$ Figure~\ref{fig:smooth} compares the columns of a single realization $X^*$ of $X$ to those of $\sqrt{p}\,Q_{X^*}$ for parameters $\nu=3, \rho=12$ and dimensions $p=100,k=3.$ The columns of the two matrices share a similar smooth structure, as we should expect based on Theorem~\ref{thm:limitingdistribution}. Proposition~\ref{prop:corr} confirms that Condition~\ref{ass:Omega} of Theorem~\ref{thm:limitingdistribution} holds for a sequence of correlation matrices constructed from the Mat\'ern correlation function \eqref{materncorrfunction}. 
\end{example} 



\begin{figure}[H]
\centerline{\includegraphics[width=\linewidth]{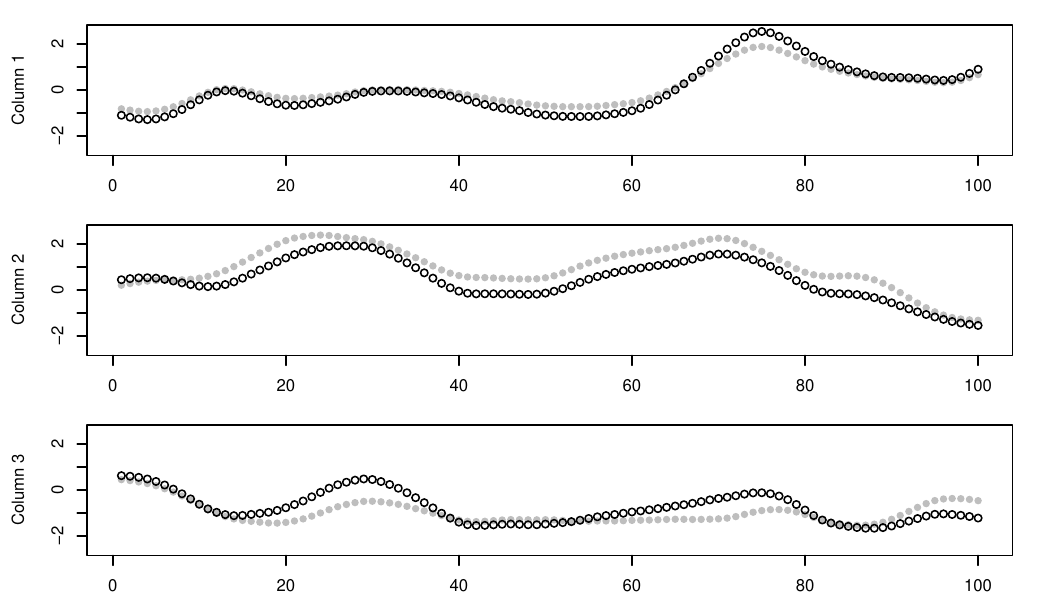}}
\caption{A comparison of the columns of a single realization $X^*$ of $X$ with those of $\sqrt{p} Q_{X^*}$ when the entries $Z$ are i.i.d. standard normals and $\Omega$ is constructed from the Mat\'ern correlation function, as discussed in Example~\ref{ex:matern}. The entries of $X^*$ appear as gray dots while the entries of $\sqrt{p}\,Q_{X^*}$ appear as black circles.}
\label{fig:smooth}
\end{figure}

\begin{proposition} \label{prop:corr}
Condition 2 of Theorem~\ref{thm:limitingdistribution} is satisfied for $\Omega = (\Omega_{ij})$ constructed in any of the following ways: 
\begin{enumerate}[label=\Roman*.,ref=\Roman*]
\item (Power correlation) $\Omega_{ij} = C^{\text{Power}}_{\rho}(|i - j|)$ where $C^{\text{Power}}_{\rho}(d) = \rho^{d}$ for $0 < \rho < 1;$  
\label{item1examples}
\item (Squared exponential correlation) $\Omega_{ij} = C^{\text{SE}}_{\rho}(|i - j|)$ where $$C^{\text{SE}}_{\rho}(d)=\exp\left[-d^2/(2\rho^2)\right]$$ for $\rho > 0;$
\label{item2examples}
\item (Mat\'ern correlation) $\Omega_{ij} = C^{\text{Mat\'ern }}_{\rho, \nu}(|i - j|)$ where $C_{\rho, \nu}$ is defined as in \eqref{materncorrfunction} for $\rho, \nu > 0;$
\label{item3examples}
\item (Banded correlation) Let $C$ be any of the correlation functions above and set
\begin{equation*}
    \Omega_{ij} = 
    \begin{cases}
        C(|i-j|)   \quad &|i-j| \leq \tau \\
        0       \quad &|i-j| > \tau
    \end{cases}
\end{equation*} for a bandwidth $0 \leq \tau < p .$
\label{item4examples}
\end{enumerate}
\end{proposition}

\begin{example} \label{ex:asymptotic_smooth}
We consider the prior for a smooth semi-orthogonal matrix described in Example~\ref{ex:matern} in the asymptotic setting of Theorem~\ref{thm:limitingdistribution}. For each $k,$ let $Z_k$ be a $p \times k$ matrix whose entries are i.i.d. standard normal random variables, and let the $(i,j)$th entry of $\Omega_p$ equal $C^{\text{Mat\'ern }}_{\rho,\nu}\left(\left|i - j \right|\right).$  Then set $X_k = \Omega_p^{1/2}Z_k$ and $Q_{X_{k}} = X_k\left(X_k^\top X_k\right)^{-1/2}.$ For concreteness, we will examine the asymptotic behavior of four entries of $\sqrt{p}Q_{X_k}$ and $X_k:$ those with row and column indices $(1,1), (2,1),$ $(2,2)$, and $(4,2).$ We can select these entries from $\sqrt{p}Q_{X_k}$ and $X_k$ using the matrix $\Delta_{\mathcal{M}_k^{(1)}, \ldots, \mathcal{M}_k^{(k)}}$ introduced in Section~\ref{se:wass_approx} by letting $\mathcal{M}_k^{(1)}=\{1,2\}$ and $\mathcal{M}_k^{(2)}=\{2,4\}$ with $\mathcal{M}_k^{(j)}=\emptyset$ for $j = 3, \ldots, k.$ The vectors $\Delta_{\mathcal{M}_k^{(1)}, \ldots, \mathcal{M}_k^{(k)}} \operatorname{vec}\left(\sqrt{p}Q_{X_{k}}\right)$ and $\Delta_{\mathcal{M}_k^{(1)}, \ldots, \mathcal{M}_k^{(k)}} \operatorname{vec}\left(X_k\right)$ contain the desired entries. The vector $\Delta_{\mathcal{M}_k^{(1)}, \ldots, \mathcal{M}_k^{(k)}} \operatorname{vec}\left(X_k\right) \stackrel{d}{=} \varphi$ where $\varphi \sim N_{4}(0_, \Sigma)$ with block diagonal covariance $$\Sigma = \Delta_{\mathcal{M}_k^{(1)}, \ldots, \mathcal{M}_k^{(k)}}\left(I_{k}\otimes\Omega_p\right)\Delta_{\mathcal{M}_k^{(1)}, \ldots, \mathcal{M}_k^{(k)}}^\top = \text{blkdiag}(\Sigma_1, \Sigma_2)$$ with blocks 
\begin{align*}
\Sigma_1= 
\begin{bmatrix}
1 & C^{\text{Mat\'ern }}_{\rho,\nu}\left(1\right) \\ 
C^{\text{Mat\'ern }}_{\rho,\nu}\left(1\right) & 1 
\end{bmatrix}, \quad
\Sigma_2= 
\begin{bmatrix}
1 & C^{\text{Mat\'ern }}_{\rho,\nu}\left(2\right) \\ 
C^{\text{Mat\'ern }}_{\rho,\nu}\left(2\right) & 1 
\end{bmatrix}.
\end{align*}
Proposition~\ref{prop:corr} establishes that Condition~\ref{ass:Omega} of Theorem~\ref{thm:limitingdistribution} is met, and it is easy to verify that Condition~\ref{ass:moments} of Theorem~\ref{thm:limitingdistribution} is met as well. Thus, Theorem~\ref{thm:limitingdistribution} implies that
\begin{align} \label{wassconvergence_smooth}
W_2\left[\operatorname{law}\left\{\Delta_{\mathcal{M}_k^{(1)}, \ldots, \mathcal{M}_k^{(k)}} \operatorname{vec}\left(\sqrt{p}Q_{X_k}\right)\right\}, \operatorname{law}\left\{\varphi\right\}\right] \to 0
\end{align}
as $k \to \infty$ and $p=p(k) \to \infty$ provided that $k/p \to 0.$ Because the Wasserstein distance metrizes weak convergence, we can conclude from \eqref{wassconvergence_smooth} that $\Delta_{\mathcal{M}_k^{(1)}, \ldots, \mathcal{M}_k^{(k)}} \operatorname{vec}\left(\sqrt{p}Q_{X_k}\right) \stackrel{d}{\to} N_{4}(0_, \Sigma)$ as $k \to \infty$ and $p=p(k) \to \infty$ provided that $k/p \to 0.$
\end{example}

\section{Data Examples} \label{DataExamples}

\subsection{Sparse network eigenmodel for protein interaction data}
\label{sparse_network_eigenmodel}



To demonstrate the proposed family of prior distributions in practice, we revisit a data example from \citet{Hoff2009} in which the author proposes a network eigenmodel for the protein interaction data of \citet{Butland2005}. This data example has also appeared in \citet{Jauch2021a}, \citet{Loyal2025}, and other articles. In contrast to these previous works, all of which considered a uniform prior distribution for the eigenvectors of the network eigenmodel, we assign the eigenvectors the sparsity-inducing prior introduced in Section~\ref{se:sparsity}. Our experiments indicate that we can learn an appropriate level of sparsity from the data, leading to more interpretable results without sacrificing predictive performance. 


The interaction data of $p=270$ proteins of \textit{Escherichia coli} are recorded in a binary, symmetric $p\times p$ matrix $Y = \left(y_{ij}\right).$ We have $y_{ij}=1$ if protein $i$ and protein $j$ interact while $y_{i,j}=0$ otherwise. The network eigenmodel assumes that 
\begin{align*}y_{ij} &\mid \pi_{ij} \stackrel{\text{ind.}}{\sim} \text{Bernoulli}\left(\pi_{ij}\right) \\ 
\pi_{ij} &=  \Phi\left\{c + \left(Q \Lambda Q^\top\right)_{ij}\right\}
\end{align*}
where $\Phi$ is the cumulative distribution function of a standard normal random variable and $(c, Q, \Lambda)$ are unknown parameters. The parameter $Q$ is a $p \times k$ semi-orthogonal matrix, $\Lambda=\text{diag}(\lambda_1, \ldots, \lambda_k)$ is a $k \times k$ diagonal matrix, and $c$ is a real number. As in \citet{Hoff2009} and \citet{Jauch2021a}, the diagonal elements of $\Lambda$ have independent $N(0,p)$ prior distributions while $c \sim N(0,10^2).$ See Section 5.1 of \citet{Jauch2021a} for a discussion of the identifiability of the network eigenmodel.


The eigenvector matrix $Q$ is assigned the sparsity-inducing prior introduced in Section~\ref{se:sparsity}. More precisely, we let $Z=(Z_{ij})$ be a $p \times k$ matrix with i.i.d. real entries distributed according to the density \eqref{zdensity}.
The prior for $Q$ is the distribution of $Q_Z = Z\left(Z^\top Z \right)^{-1/2},$ the projection of $Z$ onto the Stiefel manifold obtained via the polar decomposition. The density \eqref{zdensity} includes a parameter $\ell \in (0,1)$ that determines the sparsity of the eigenvector matrix $Q,$ with a smaller value of $\ell$ corresponding to greater sparsity. Instead of choosing a fixed value, we assign $\ell \in (0,1)$ a uniform prior and learn it from the data. In practice, we use the scale mixture representation given in Proposition~\ref{scalemixture} so that $Z_{ij} \, \vert\,  \theta_{ij} , \ell \stackrel{\text{ind.}}{\sim} \text{Normal}\left(0, \theta_{ij}/\ell\right)$ where $\theta_{ij} \stackrel{\text{i.i.d.}}{\sim} \text{Beta}\left[\ell/2, \left(1-\ell\right)/2\right].$ 
Let $\Theta$ denote the $p \times k$ matrix whose $ij$th entry is $\theta_{ij}.$


We obtain an MCMC approximation to the posterior distribution by applying polar expansion and adaptive HMC as implemented in Stan. After polar expansion, the posterior density is 
\begin{align}
    p\left(c, Z, \Theta, \Lambda, \ell \mid Y \right) &\propto \prod_{1\leq i < j \leq n}
   \Phi\left\{c + \left( Q_Z \Lambda Q_Z^\top\right)_{ij}\right\}^{y_{ij}} \left[1- \Phi\left\{c + \left( Q_Z \Lambda Q_Z^\top\right)_{ij}\right\}\right]^{1-y_{ij}} \nonumber \\
    &\times \prod_{i=1}^p \prod_{j=1}^k \text{Normal}\left(z_{ij} \mid 0, \theta_{ij}/\ell\right) \text{Beta}\left[\theta_{ij} \mid \ell/2, \left(1-\ell\right)/2\right] \nonumber \\
    & \times \text{Normal}\left(c \mid 0, 10^2\right) \times \mathbbm{1}_{(0,1)}(\ell) \times \prod_{j=1}^{k} \text{Normal}\left(\lambda_j \mid 0, p\right). \label{network_eigen_post}
\end{align}
Given a Markov chain $\{c_t, Z_t, \Theta_t, \Lambda_t, \ell_t \}_{t=1}^T$ whose stationary distribution has density \eqref{network_eigen_post}, we approximate the marginal posterior distribution of $Q$ by $\{ Q_t\}_{t=1}^T$ where $Q_t = Z_t \left( Z_t^\top Z_t\right)^{-1/2}$ for each $t.$


Previous analyses of the protein interaction data based on the network eigenmodel have differed in their choice of dimension $k.$ \citet{Hoff2009} and \citet{Jauch2021a} chose $k=3$ without much discussion. The dimension selection procedure of \citet{Loyal2025} led to a choice of $k=5$ which was further supported by a posterior predictive analysis. In our experiments, we set $k=5.$


\begin{figure}%
    \centering
    \subfloat[\centering ]{{\includegraphics[width=6.5cm]{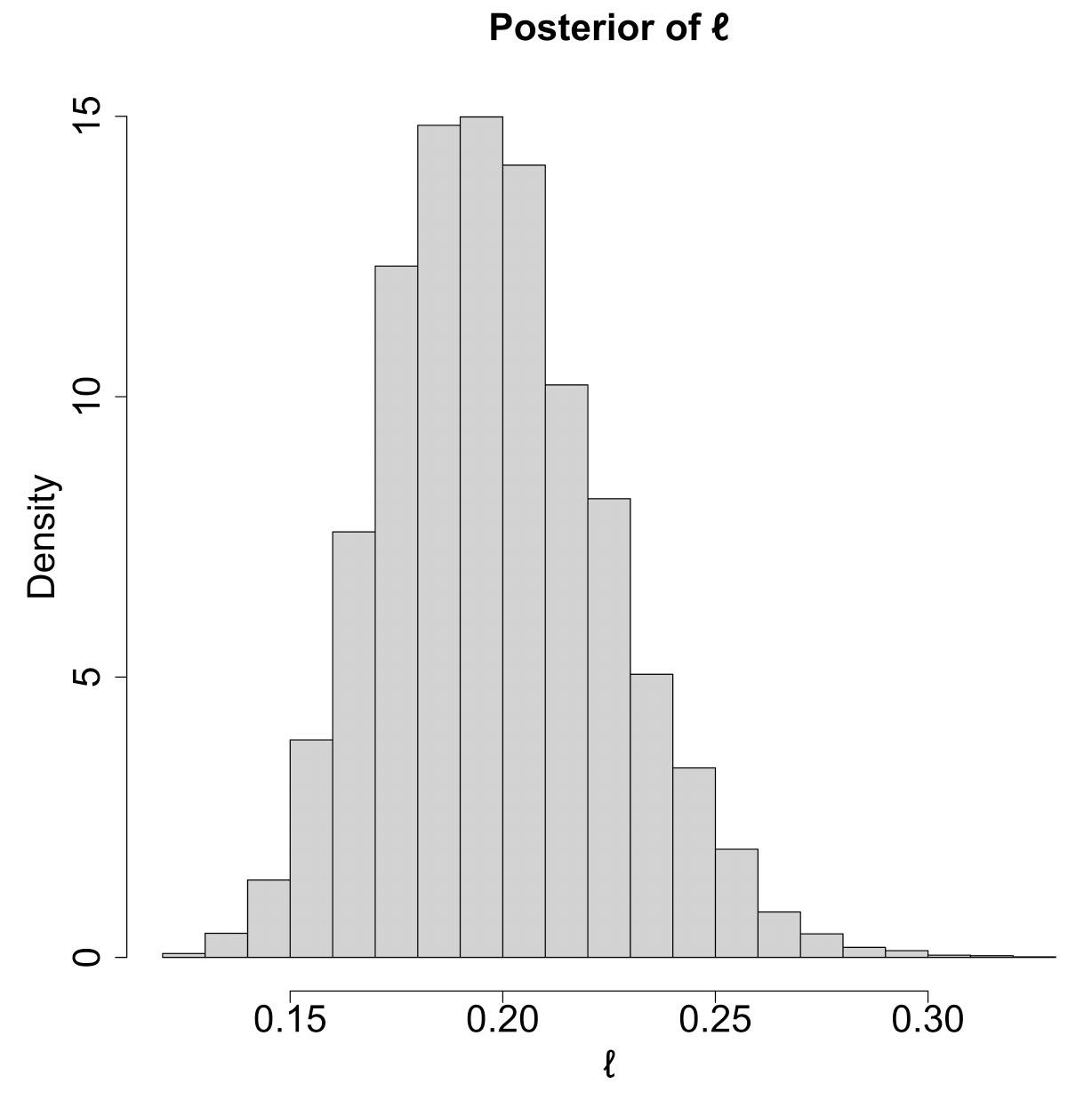} }}%
    \subfloat[\centering ]{{\includegraphics[width=6.5cm]{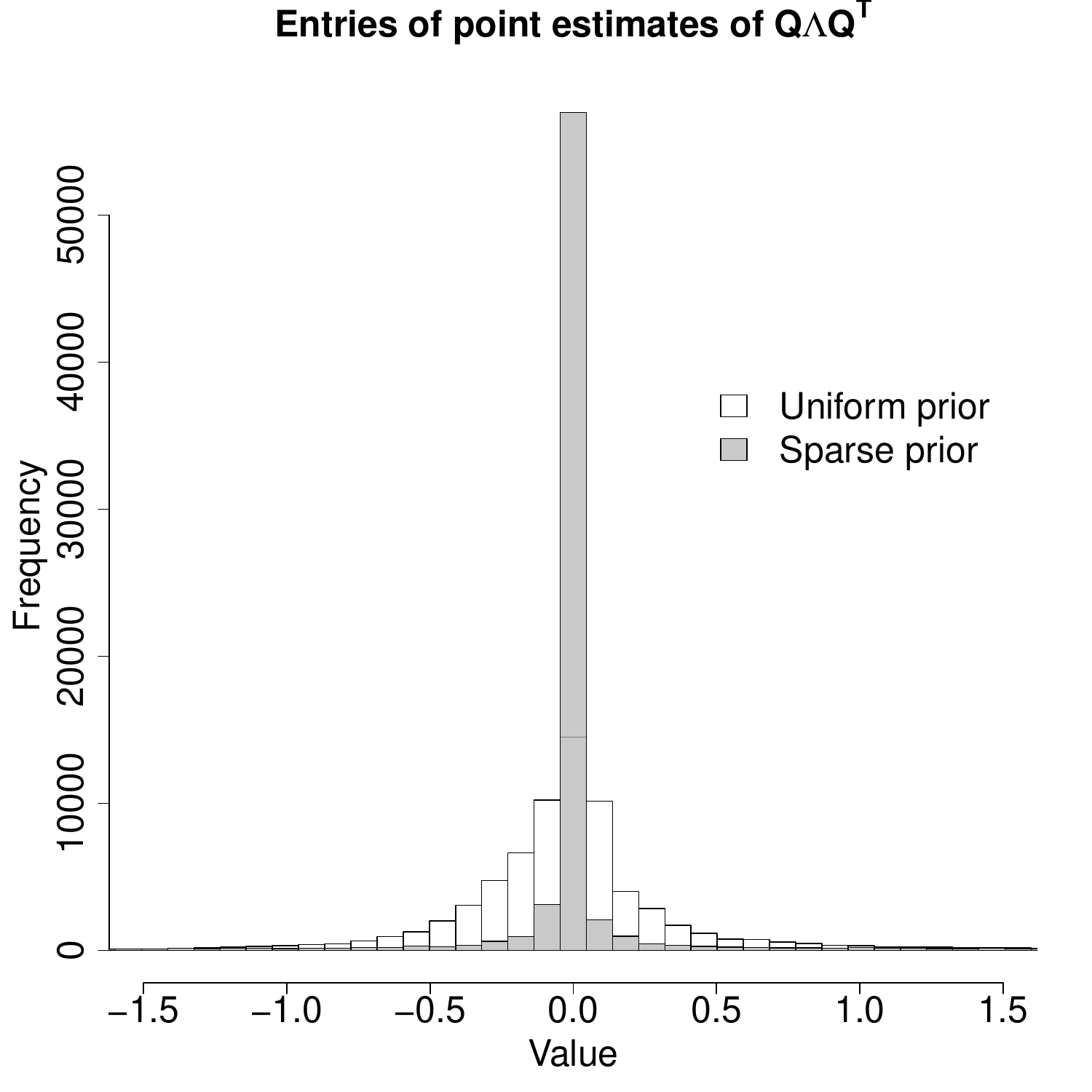} }}%
    \caption{Panel (a) shows the marginal posterior distribution of the sparsity parameter $\ell \in (0,1).$ A smaller value of $\ell$ leads to greater sparsity. Panel (b) compares the entries of the point estimate of $Q \Lambda Q^T$ under the two  priors for $Q.$}%
    \label{fig:network_eigen}%
\end{figure}


In the first experiment, we fit both the sparse network eigenmodel (with the sparsity-inducing prior for $Q$ described above) and the standard network eigenmodel (with the uniform prior for $Q$) to the full protein network data set. 
Figure~\ref{fig:network_eigen} provides posterior summary plots. Panel (a) shows the marginal posterior distribution of the sparsity parameter $\ell.$ The posterior mean is around $.20$ while $(0.15, 0.25)$ is an approximate 95\% credible interval. The small estimated value of $\ell$ indicates that the data favor a sparse $Q.$ Panel (b) compares the entries of our point estimate of $Q \Lambda Q^T$ (the element-wise posterior median) under the two priors. 
As we might expect, the point estimate of $Q \Lambda Q^T$ under the sparsity-inducing prior has many nearly zero entries with a smaller number of very large entries. This leads to greater interpretability: Nearly zero entries of $Q \Lambda Q^T$ correspond to pairs of proteins whose probability of interaction is close to the baseline value of $\Phi(c),$ while large entries correspond to pairs of proteins whose probability of interaction deviates from that baseline value.


In the second experiment, we compared the out-of-sample predictive performance of the sparse network eigenmodel to that of the standard network eigenmodel. We repeatedly fit both models to randomly selected subsets of the dyads in the protein network data, estimated the probabilities that connections exist for the held out dyads, then calculated the area under the curve (AUC) for the two models. The procedure was repeated 16 times with 1/3 of the dyads held out each time. We took posterior means as our point estimates for the probabilities. We found that the predictive performance of the sparse network eigenmodel was comparable to that of the standard network eigenmodel. In particular, the AUC of the sparse network eigenmodel was greater in 9/16 of the repetitions with a median difference of .5\% in favor of the sparse network eigenmodel. In light of its greater interpretability and comparable predictive performance, the sparse network eigenmodel is an appealing alternative to the standard network eigenmodel for this protein interaction data. 


\subsection{Smooth PCA of Hawaii ocean oxygen data}


The \href{https://hahana.soest.hawaii.edu/hot/}{Hawaii Ocean Time-series (HOT) program} has measured physical and biogeochemical variables at a deep-water sampling site 100 km north of Oahu since October 1988. The measurements are made across a range of depths on an approximately monthly basis. Long-term time series like those collected by HOT are valuable in climate studies, because they help researchers understand variability and changes in physical and biogeochemical processes in the ocean. 

In this section, we fit the model-based SVD from Example~\ref{ex:SVD} to oxygen concentration data collected by the HOT program and analyzed in \citet{zhang2023robustbayesianfunctionalprincipal}. To account for the functional relationship between oxygen concentration and depth, we assign $V$ a prior appropriate for a semi-orthogonal matrix with smooth structure, as discussed in Section~\ref{se:smoothness}. The model and prior specification are similar to those from the PCA of the Canadian weather data in Section 5.2 of \citet{Jauch2021a}. That analysis can be viewed as another illustration of the proposed family of prior distributions. The oxygen concentration dataset can be obtained from the \href{https://hahana.soest.hawaii.edu/hot/hot-dogs/cextraction.html}{Hawaii Ocean Time-series Data Organization \& Graphical System (HOT-DOGS)} or the \href{https://github.com/SFU-Stat-ML/RBFPCA}{Github repository} associated with \citet{zhang2023robustbayesianfunctionalprincipal}.


The data 
include oxygen concentrations in units of $\mu$mol/kg observed on $n=133$ separate occasions from January 1, 2008 to December 31, 2021. On each occasion, oxygen concentrations were measured at $p=100$ evenly spaced depths ranging from 2 meters to 200 meters below the ocean surface. 
The raw data matrix $Y_\text{raw}$ has $n=133$ rows and $p=100$ columns with entry $(i,j)$ recording the oxygen concentration at a depth of $2j$ meters on the $i$th occasion. 


We subtract column means from $Y_\text{raw}$ and model the resulting matrix as $Y = UDV^\top + \sigma E.$ The semi-orthogonal matrices $U \in \mathcal{V}(k,n)$ and $V \in \mathcal{V}(k,p),$ the diagonal matrix $D = \text{diag}(d_1, \ldots, d_k)$ with $d_1, \ldots, d_k > 0,$ and the scale $\sigma > 0$ are unknown parameters. The entries of $E$ are independent standard normal errors. As in Section 5.2 of \citet{Jauch2021a}, we can view this model-based SVD as a PCA of the functional data. The rows of $UD$ contain the principal component scores for each observation, while the columns of $V$ form the corresponding basis of principal component curves. In this analysis, we set $k=4.$


We assign the semi-orthogonal matrix $V$ a hierarchical prior chosen to reflect the functional nature of the oxygen concentration data. In particular, we let 
\begin{align*}
V \mid \rho &\sim \text{MACG}\left(\Omega\right) \\ 
1/\rho &\sim \text{Gamma}(\alpha, \beta)
\end{align*}
where $\Omega = (\Omega_{ij})$ is constructed from the squared exponential correlation function from Proposition~\ref{prop:corr} with $\Omega_{ij} = C^{\text{SE}}_{\rho}(|2i - 2j|).$ The squared exponential correlation function can be obtained from the Mat\'ern correlation function of Example~\ref{ex:matern} by letting $\nu \to \infty.$ The length-scale parameter $\rho,$ which controls the ``wiggliness" of the principal component curves, is treated as unknown and assigned an inverse gamma prior. The inverse gamma prior places negligible probability mass near zero, effectively ruling out values of $\rho$ that are too small to make scientific sense. We select the shape $\alpha$ and the rate $\beta$ of the inverse gamma prior according to the heuristic of Section 5.2 of \citet{Jauch2021a}, which we review here. For a fixed $\rho,$ we can reason (based on Theorem~\ref{thm:limitingdistribution}, Proposition~\ref{prop:corr}, and the fact that $p \gg k$) that \textit{a priori} each column of $V$ will resemble a mean-zero Gaussian process (GP) with a squared exponential correlation function and length-scale $\rho.$ The expected number of zero crossings by such a GP within the interval $[0,T]$ is $T/(\pi\rho)$ \citep{rice1945mathematical, Kratz2006, Rasmussen2006}. \footnote{The formula stated in \citet{Jauch2021a} is off by a factor of two.} 
We choose $\alpha$ and $\beta$ so that $\rho$ has a prior mean of $200/(2\pi)$ and a prior standard deviation (SD) of 10. If $\rho$ were fixed at its prior mean, the expected number of zero crossings of the principal component curves would be approximately two. 


We now specify the priors for the remaining parameters $U, \sigma^2,$ and $d_1, \ldots, d_k.$ We assign the semi-orthogonal matrix $U$ a uniform prior. The priors for $\sigma^2$ and $d_1, \ldots, d_k$ are inverse gamma and truncated normal: 
\begin{align*}
    1/\sigma^2 &\sim \text{Gamma}\left(\frac{\nu}{2},\frac{\nu}{2} s^2\right) \\ 
    p(d_1, \ldots , d_k) &\propto \mathbbm{1}\left\{d_1, \ldots , d_k > 0\right\} \prod_{i=1}^k \text{Normal}(d_i \mid 0, \tau^2).
\end{align*} 
We select the hyperparameters $\nu, s^2,$ and $\tau$ using the empirical Bayes strategy from Section 5.2 of \citet{Jauch2021a}. 


We simulate from the posterior distribution via polar expansion combined with adaptive HMC, as described in Section~\ref{sec:preliminaries}. As our point estimate of $V,$ we take the first $k=4$ right singular vectors of the posterior mean of $UDV^\top.$ The left side of Figure~\ref{fig:hawaii} compares our point estimate (in black) to the results of classical PCA (in gray). The principal component curves produced by our approach are slightly smoother than those obtained from classical PCA but are overall very similar. (The differences are more pronounced in the analysis of the Canadian weather data in Section 5.2 of \citet{Jauch2021a}.) The estimated principal component curves resemble those from \citet{zhang2023robustbayesianfunctionalprincipal} and lead to similar interpretations. The first PC relates to the overall oxygen concentration level across the entire range of depths, with a higher PC score corresponding to a higher concentration. The second PC relates to differences in oxygen concentrations at middle depths (approximately 50-125 meters) compared to deeper depths (approximately 150-200 meters), with a higher PC score corresponding to increased oxygen concentration at middle depths. The third and fourth PCs relate to more complex phenomena. The right side of Figure~\ref{fig:hawaii} compares a histogram estimate of the posterior density of $\rho$ with its inverse gamma prior density. Compared to the prior (mean $31.83,$ SD $10$), the posterior of $\rho$ is much more concentrated with with a slightly lower mean (mean $25.13,$ SD $1.41$), indicating that we can learn a suitable length-scale $\rho$ from the data. 

\begin{figure}[H]
    \centering 
    \subfloat[\centering ]{{\includegraphics[width=6.5cm]{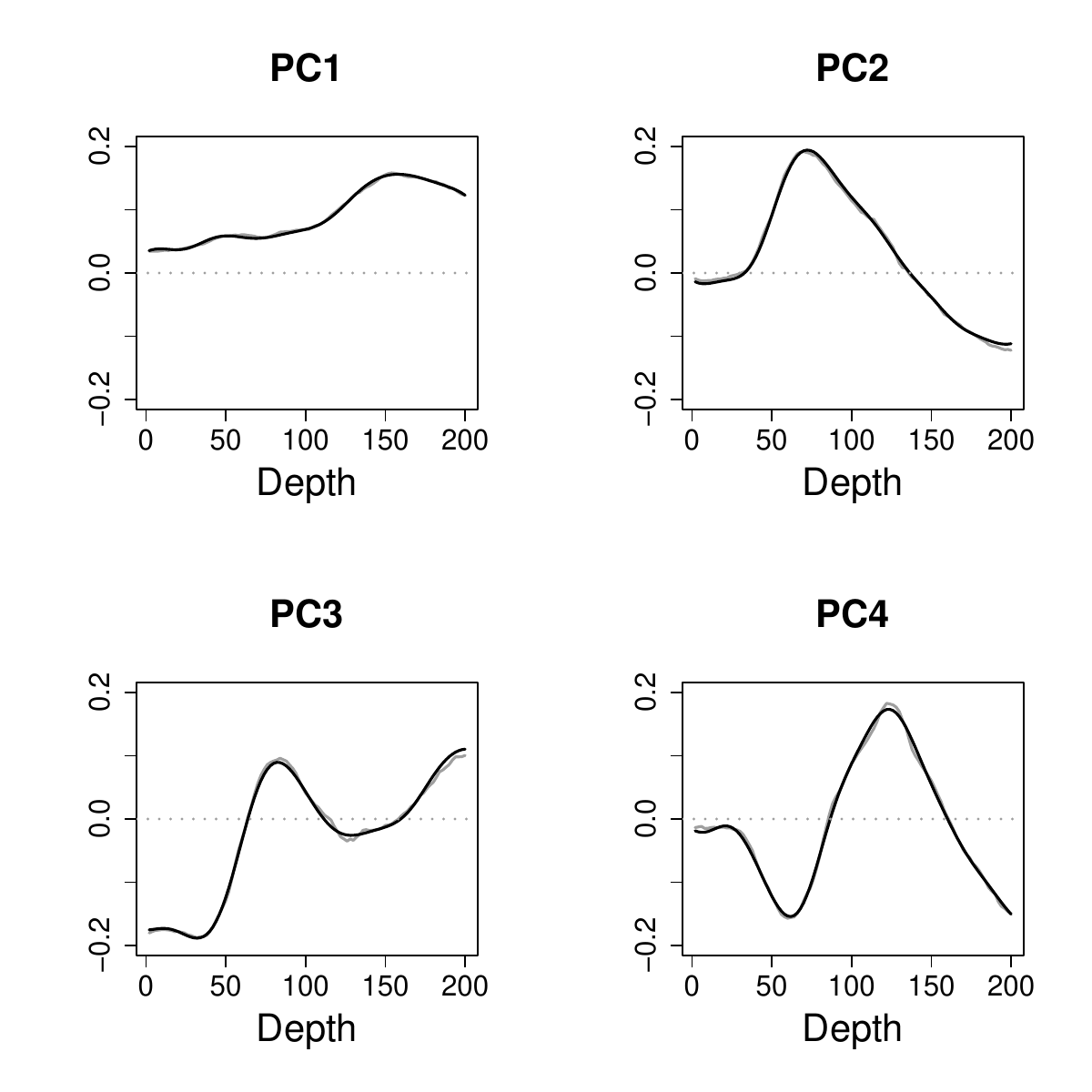} }}
    \hspace{-.1cm}
    \subfloat[\centering ]{{\includegraphics[width=6.5cm]{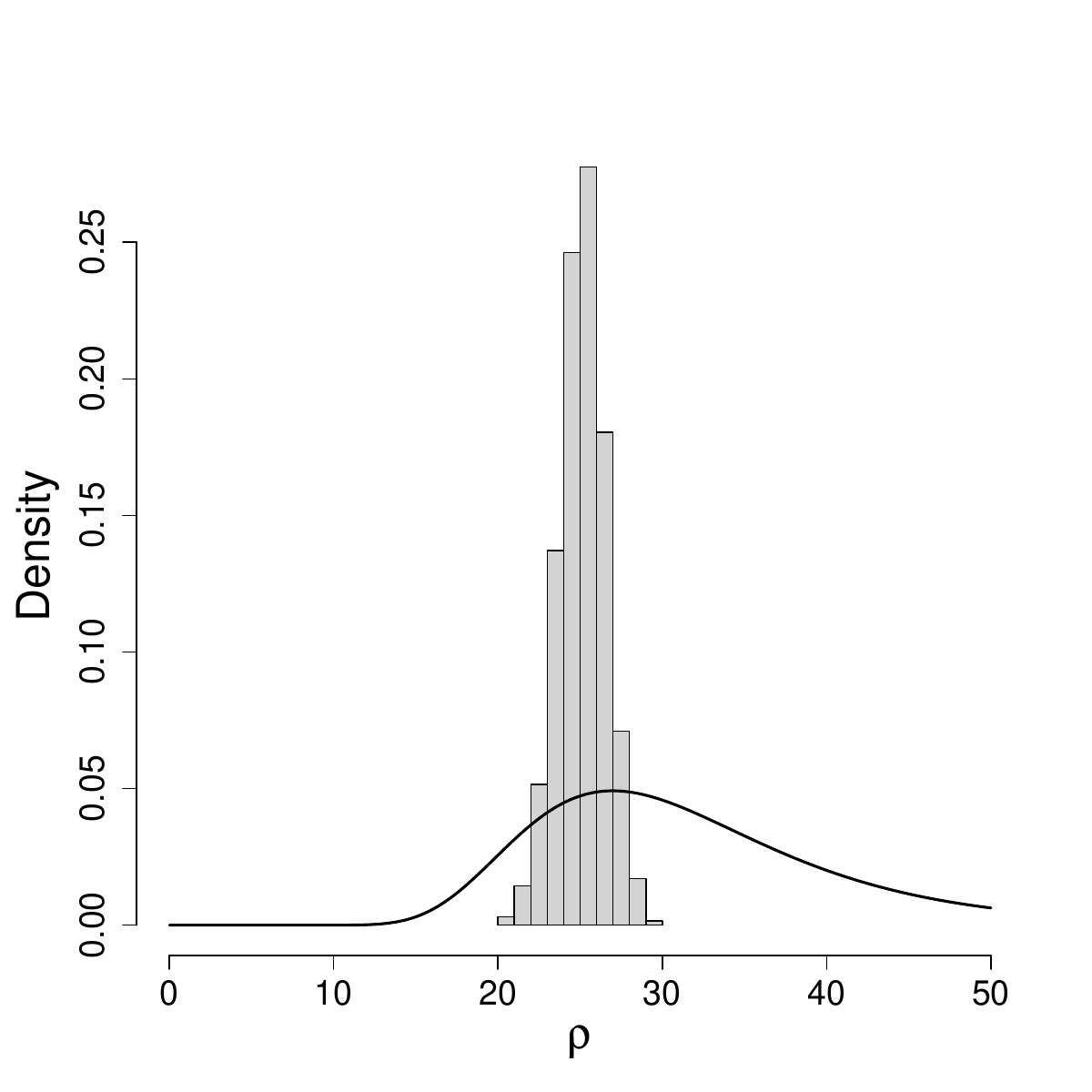} }}
    \caption{The left side compares our point estimate (in black) to the results of classical PCA (in gray). The right side compares a histogram estimate of the posterior density of $\rho$ with its inverse gamma prior density.}%
    \label{fig:hawaii}
\end{figure}

\section{Discussion}

We presented an approach to defining prior distributions for structured semi-orthogonal matrices that leads to tractable posterior inference via parameter expanded MCMC. The key idea was to introduce a suitably constructed latent random matrix $X$ whose projection $Q_X$ onto the Stiefel manifold is identified with the semi-orthogonal matrix parameter. We demonstrated that $Q_X$ inherits the invariance properties of $X$ (Theorem~\ref{prop:invariance}) and established that the Wasserstein distance between the distribution of elements of $\sqrt{p}Q_X$ and the distribution of the corresponding elements of $X$ is small when $k \ll p$ (Theorem~\ref{thm:limitingdistribution}). These theoretical results, which may be of independent interest to those working in high-dimensional probability and random matrix theory, imply that we can incorporate structural assumptions into the prior distribution for a semi-orthogonal matrix by building these features into the distribution of $X.$ To incorporate sparsity, we introduced a novel shrinkage prior with a number of appealing properties, including a representation as a scale mixture of normal distributions. To incorporate smoothness, we considered a variety of correlation functions. We illustrated these priors in applications to protein interaction data and ocean oxygen data.

There are several interesting directions for future research: 

\begin{itemize}
\item The theoretical contributions of this article were focused on properties of the proposed family of prior distributions. 
A natural next step would be to investigate properties of the resulting posterior distribution in specific models like the model-based SVD of Example~\ref{ex:SVD}, the network eigenmodel of Example~\ref{ex:eigenmodel}, or the spiked covariance model \citep{Cai2013, Gao2015, Xie2024} under the assumption that a semi-orthogonal matrix parameter is sparse or smooth.

\item In some settings, it is natural to assume that a semi-orthogonal matrix parameter is \textit{jointly sparse} with entire rows of zeros \citep{Vu2013, cai2015optimal, Xie2022} or has heterogeneous structure with different levels of sparsity or smoothness across its columns \citep{north2024flexible}. The construction of Section~\ref{sec:the_family_and_properties}, which requires that the columns of $X$ are i.i.d., cannot easily accommodate these structural assumptions. Another natural next step would be to generalize that construction and develop analogous theory with an eye toward addressing this limitation.

\item In applications involving structured semi-orthogonal matrix parameters, there is almost always uncertainty regarding the number of columns $k.$  To account for this uncertainty, we could incorporate the model averaging and dimension selection ideas proposed in \citet{Hoff2007} or \citet{Loyal2025}. 

\item An interesting and challenging mathematical problem would be to investigate the distance between the distributions of $\sqrt{p} Q_X$ and $X$ (or transformations thereof) in terms of other probability metrics \citep{gibbs2002choosing}. 

\end{itemize}

\begin{funding}
The first author was supported by the U.S. National Science Foundation (DMS-2515376). 

\end{funding}

\begin{supplement}
The \href{https://michaeljauch.github.io/SupplementPriorsPaper.pdf}{supplemental material} includes proofs of all theorems and propositions, as well as the additional proposition described at the end of Section~\ref{sec:the_family_and_properties} 
\end{supplement}

\bibliographystyle{ba}
\bibliography{biblio,references}

\begin{thebibliography}{48}
\newcommand{\enquote}[1]{``#1''}
\expandafter\ifx\csname natexlab\endcsname\relax\def\natexlab#1{#1}\fi
\expandafter\ifx\csname url\endcsname\relax
  \def\url#1{{\tt #1}}\fi
\expandafter\ifx\csname urlprefix\endcsname\relax\def\urlprefix{URL }\fi
\ifx\endbibitem\undefined \let\endbibitem\relax\fi

\bibitem[{Armagan et~al.(2013)Armagan, Dunson, and Lee}]{armagan2013generalized}
Armagan, A., Dunson, D.~B., and Lee, J. (2013).
\newblock \enquote{Generalized double {P}areto shrinkage.}
\newblock {\em Statistica Sinica\/}, 23(1): 119.
\endbibitem

\bibitem[{Bae and Mallick(2004)}]{Bae2004}
Bae, K. and Mallick, B.~K. (2004).
\newblock \enquote{Gene selection using a two-level hierarchical {B}ayesian model.}
\newblock {\em Bioinformatics\/}, 20(18): 3423--3430.
\endbibitem

\bibitem[{Bhattacharya et~al.(2015)Bhattacharya, Pati, Pillai, and Dunson}]{Bhattacharya2015}
Bhattacharya, A., Pati, D., Pillai, N.~S., and Dunson, D.~B. (2015).
\newblock \enquote{Dirichlet–{L}aplace Priors for Optimal Shrinkage.}
\newblock {\em Journal of the American Statistical Association\/}, 110(512): 1479--1490.
\endbibitem

\bibitem[{Bryan et~al.(2021)Bryan, Niles-Weed, and Hoff}]{bryan2021multirank}
Bryan, J.~G., Niles-Weed, J., and Hoff, P.~D. (2021).
\newblock \enquote{The multirank likelihood for semiparametric canonical correlation analysis.}
\newblock {\em arXiv:2112.07465\/}.
\endbibitem

\bibitem[{Butland et~al.(2005)Butland, Peregrín-Alvarez, Li, Yang, Yang, Canadien, Starostine, Richards, Beattie, Krogan, Davey, Parkinson, Greenblatt, and Emili}]{Butland2005}
Butland, G., Peregrín-Alvarez, J.~M., Li, J., Yang, W., Yang, X., Canadien, V., Starostine, A., Richards, D., Beattie, B., Krogan, N., Davey, M., Parkinson, J., Greenblatt, J., and Emili, A. (2005).
\newblock \enquote{Interaction network containing conserved and essential protein complexes in {E}scherichia coli.}
\newblock {\em Nature\/}, 531–537.
\endbibitem

\bibitem[{Cai et~al.(2015)Cai, Ma, and Wu}]{cai2015optimal}
Cai, T., Ma, Z., and Wu, Y. (2015).
\newblock \enquote{Optimal estimation and rank detection for sparse spiked covariance matrices.}
\newblock {\em Probability Theory and Related Fields\/}, 161(3): 781--815.
\endbibitem

\bibitem[{Cai et~al.(2013)Cai, Ma, and Wu}]{Cai2013}
Cai, T.~T., Ma, Z., and Wu, Y. (2013).
\newblock \enquote{{Sparse PCA: Optimal rates and adaptive estimation}.}
\newblock {\em The Annals of Statistics\/}, 41(6): 3074--3110.
\endbibitem

\bibitem[{Carpenter et~al.(2017)Carpenter, Gelman, Hoffman, Lee, Goodrich, Betancourt, Brubaker, Guo, Li, and Riddell}]{Carpenter2017}
Carpenter, B., Gelman, A., Hoffman, M.~D., Lee, D., Goodrich, B., Betancourt, M., Brubaker, M., Guo, J., Li, P., and Riddell, A. (2017).
\newblock \enquote{Stan: a probabilistic programming language.}
\newblock {\em Journal of Statistical Software\/}, 76(1): 1--32.
\endbibitem

\bibitem[{Carvalho et~al.(2010)Carvalho, Polson, and Scott}]{carvalho2010horseshoe}
Carvalho, C.~M., Polson, N.~G., and Scott, J.~G. (2010).
\newblock \enquote{The horseshoe estimator for sparse signals.}
\newblock {\em Biometrika\/}, 97(2): 465--480.
\endbibitem

\bibitem[{Chikuse(1990)}]{Chikuse1990}
Chikuse, Y. (1990).
\newblock \enquote{The Matrix Angular Central {G}aussian Distribution.}
\newblock {\em Journal of Multivariate Analysis\/}, 33(2): 265--274.
\endbibitem

\bibitem[{Chikuse(2003)}]{Chikuse2003}
--- (2003).
\newblock {\em Statistics on Special Manifolds\/}.
\newblock Springer New York.
\endbibitem

\bibitem[{Choi et~al.(2017)Choi, Taylor, and Tibshirani}]{choi2017selecting}
Choi, Y., Taylor, J., and Tibshirani, R. (2017).
\newblock \enquote{Selecting the number of principal components: Estimation of the true rank of a noisy matrix.}
\newblock {\em The Annals of Statistics\/}, 45(6): 2590--2617.
\endbibitem

\bibitem[{Dawid(1981)}]{Dawid1981}
Dawid, A.~P. (1981).
\newblock \enquote{Some matrix-variate distribution theory: Notational considerations and a {B}ayesian application.}
\newblock {\em Biometrika\/}, 68(1): 265--274.
\endbibitem

\bibitem[{Figueiredo(2003)}]{Figueiredo2003}
Figueiredo, M. (2003).
\newblock \enquote{Adaptive sparseness for supervised learning.}
\newblock {\em IEEE Transactions on Pattern Analysis and Machine Intelligence\/}, 25(9): 1150--1159.
\endbibitem

\bibitem[{Franks and Hoff(2019)}]{Franks2019}
Franks, A.~M. and Hoff, P.~D. (2019).
\newblock \enquote{Shared Subspace Models for Multi-Group Covariance Estimation.}
\newblock {\em Journal of Machine Learning Research\/}, 20(171): 1--37.
\endbibitem

\bibitem[{Gao and Zhou(2015)}]{Gao2015}
Gao, C. and Zhou, H.~H. (2015).
\newblock \enquote{Rate-optimal posterior contraction for sparse {PCA}.}
\newblock {\em The Annals of Statistics\/}, 43(2): 785--818.
\endbibitem

\bibitem[{Gavish and Donoho(2017)}]{gavish2017optimal}
Gavish, M. and Donoho, D.~L. (2017).
\newblock \enquote{Optimal shrinkage of singular values.}
\newblock {\em IEEE Transactions on Information Theory\/}, 63(4): 2137--2152.
\endbibitem

\bibitem[{Gibbs and Su(2002)}]{gibbs2002choosing}
Gibbs, A.~L. and Su, F.~E. (2002).
\newblock \enquote{On choosing and bounding probability metrics.}
\newblock {\em International Statistical Review\/}, 70(3): 419--435.
\endbibitem

\bibitem[{Higham and Schreiber(1990)}]{higham1990fast}
Higham, N.~J. and Schreiber, R.~S. (1990).
\newblock \enquote{Fast polar decomposition of an arbitrary matrix.}
\newblock {\em SIAM Journal on Scientific and Statistical Computing\/}, 11(4): 648--655.
\endbibitem

\bibitem[{Hoff(2007)}]{Hoff2007}
Hoff, P.~D. (2007).
\newblock \enquote{Model averaging and dimension selection for the singular value decomposition.}
\newblock {\em Journal of the American Statistical Association\/}, 102(478): 674--685.
\endbibitem

\bibitem[{Hoff(2009)}]{Hoff2009}
--- (2009).
\newblock \enquote{Simulation of the matrix {B}ingham-von {M}ises-{F}isher Distribution, with applications to multivariate and relational data.}
\newblock {\em Journal of Computational and Graphical Statistics\/}, 18(2): 438--456.
\endbibitem

\bibitem[{Hoff(2016)}]{Hoff2016}
--- (2016).
\newblock \enquote{{Equivariant and Scale-Free Tucker Decomposition Models}.}
\newblock {\em Bayesian Analysis\/}, 11(3): 627--648.
\endbibitem

\bibitem[{Hoffman and Gelman(2014)}]{Hoffman2014}
Hoffman, M.~D. and Gelman, A. (2014).
\newblock \enquote{The {N}o-{U}-{T}urn Sampler: Adaptively Setting Path Lengths in {H}amiltonian {M}onte {C}arlo.}
\newblock {\em Journal of Machine Learning Research\/}, 15: 1351--1381.
\endbibitem

\bibitem[{Jauch et~al.(2020)Jauch, Hoff, and Dunson}]{Jauch2020}
Jauch, M., Hoff, P.~D., and Dunson, D.~B. (2020).
\newblock \enquote{Random orthogonal matrices and the {C}ayley transform.}
\newblock {\em Bernoulli\/}, 26: 1560--1586.
\endbibitem

\bibitem[{Jauch et~al.(2021)Jauch, Hoff, and Dunson}]{Jauch2021a}
--- (2021).
\newblock \enquote{Monte {C}arlo simulation on the {S}tiefel manifold via polar expansion.}
\newblock {\em Journal of Computational and Graphical Statistics\/}, 30: 622--631.
\endbibitem

\bibitem[{Khatri and Mardia(1977)}]{Khatri1977}
Khatri, C.~G. and Mardia, K.~V. (1977).
\newblock \enquote{The Von {M}ises-{F}isher Matrix Distribution in Orientation Statistics.}
\newblock {\em Journal of the Royal Statistical Society. Series B (Methodological)\/}, 39(1): 95--106.
\endbibitem

\bibitem[{Kratz(2006)}]{Kratz2006}
Kratz, M.~F. (2006).
\newblock \enquote{{Level crossings and other level functionals of stationary Gaussian processes}.}
\newblock {\em Probability Surveys\/}, 3(none): 230--288.
\endbibitem

\bibitem[{Li et~al.(2024)Li, Choi, Dunlop, Craighead, Mayberg, Garmire, Guo, and Kang}]{li2024bsnmani}
Li, Y., Choi, K.~S., Dunlop, B.~W., Craighead, W.~E., Mayberg, H.~S., Garmire, L., Guo, Y., and Kang, J. (2024).
\newblock \enquote{BSNMani: Bayesian Scalar-on-network Regression with Manifold Learning.}
\newblock {\em arXiv:2410.02965\/}.
\endbibitem

\bibitem[{Loyal and Chen(2025)}]{Loyal2025}
Loyal, J.~D. and Chen, Y. (2025).
\newblock \enquote{A spike-and-slab prior for dimension selection in generalized linear network eigenmodels.}
\newblock {\em Biometrika\/}, 112(3): asaf014.
\endbibitem

\bibitem[{Meng and Bouchard(2024)}]{meng2024bayesian}
Meng, R. and Bouchard, K.~E. (2024).
\newblock \enquote{Bayesian inference of structured latent spaces from neural population activity with the orthogonal stochastic linear mixing model.}
\newblock {\em PLOS Computational Biology\/}, 20(4): e1011975.
\endbibitem

\bibitem[{Neal(2011)}]{Neal2011}
Neal, R.~M. (2011).
\newblock {\em {MCMC} using {H}amiltonian dynamics\/}, 113–162.
\newblock CRC Press.
\endbibitem

\bibitem[{North et~al.(2024)North, Risser, and Breidt}]{north2024flexible}
North, J.~S., Risser, M.~D., and Breidt, F.~J. (2024).
\newblock \enquote{A flexible class of priors for orthonormal matrices with basis function-specific structure.}
\newblock {\em Spatial Statistics\/}, 64: 100866.
\endbibitem

\bibitem[{Park(2024)}]{Park2024}
Park, H.~G. (2024).
\newblock \enquote{Bayesian estimation of covariate assisted principal regression for brain functional connectivity.}
\newblock {\em Biostatistics\/}, kxae023.
\endbibitem

\bibitem[{Pourzanjani et~al.(2021)Pourzanjani, Jiang, Mitchell, Atzberger, and Petzold}]{Pourzanjani2021}
Pourzanjani, A.~A., Jiang, R.~M., Mitchell, B., Atzberger, P.~J., and Petzold, L.~R. (2021).
\newblock \enquote{Bayesian Inference over the {S}tiefel {M}anifold via the {G}ivens Representation.}
\newblock {\em Bayesian Analysis\/}, 16: 39–666.
\endbibitem

\bibitem[{Rasmussen and Williams(2006)}]{Rasmussen2006}
Rasmussen, C.~E. and Williams, C. K.~I. (2006).
\newblock {\em Gaussian processes for Machine learning\/}.
\newblock MIT Press.
\endbibitem

\bibitem[{Rice(1945)}]{rice1945mathematical}
Rice, S.~O. (1945).
\newblock \enquote{Mathematical analysis of random noise.}
\newblock {\em The Bell System Technical Journal\/}, 24(1): 46--156.
\endbibitem

\bibitem[{Shepard et~al.(2015)Shepard, Brozell, and Gidofalvi}]{Shepard2015}
Shepard, R., Brozell, S.~R., and Gidofalvi, G. (2015).
\newblock \enquote{The representation and parametrization of orthogonal matrices.}
\newblock {\em The Journal of Physical Chemistry A\/}, 119: 7924--7939.
\endbibitem

\bibitem[{Srivastava and Khatri(1979)}]{Srivastava1979}
Srivastava, M.~S. and Khatri, C.~G. (1979).
\newblock {\em An Introduction to Multivariate Statistics\/}.
\newblock North-Holland/New York.
\endbibitem

\bibitem[{Stam(1982)}]{Stam1982}
Stam, A.~J. (1982).
\newblock \enquote{{Limit theorems for uniform distributions on spheres in high-dimensional Euclidean spaces}.}
\newblock {\em Journal of Applied Probability\/}, 19: 221--228.
\endbibitem

\bibitem[{Stouffer et~al.(2021)Stouffer, Godoy, Dalla~Riva, and Mayfield}]{Stouffer2021}
Stouffer, D.~B., Godoy, O., Dalla~Riva, G.~V., and Mayfield, M.~M. (2021).
\newblock \enquote{The dimensionality of plant{\textendash}plant competition.}
\newblock {\em bioRxiv\/}.
\endbibitem

\bibitem[{Villani(2009)}]{Villani2009}
Villani, C. (2009).
\newblock {\em {Optimal transport: old and new}\/}.
\newblock Springer-Verlag.
\endbibitem

\bibitem[{Vu and Lei(2013)}]{Vu2013}
Vu, V.~Q. and Lei, J. (2013).
\newblock \enquote{{Minimax sparse principal subspace estimation in high dimensions}.}
\newblock {\em The Annals of Statistics\/}, 41(6): 2905--2947.
\endbibitem

\bibitem[{Watson(1983)}]{Watson1983}
Watson, G.~S. (1983).
\newblock \enquote{{Limit theorems on high dimensional spheres and Stiefel manifolds}.}
\newblock In {\em Studies in Econometrics, Time Series, and Multivariate Statistics\/}, 559--570. Academic Press.
\endbibitem

\bibitem[{Xie(2024)}]{Xie2024}
Xie, F. (2024).
\newblock \enquote{{Spectral norm posterior contraction in Bayesian sparse spiked covariance matrix model}.}
\newblock {\em Electronic Journal of Statistics\/}, 18(2): 5198--5257.
\endbibitem

\bibitem[{Xie et~al.(2022)Xie, Cape, Priebe, and Xu}]{Xie2022}
Xie, F., Cape, J., Priebe, C.~E., and Xu, Y. (2022).
\newblock \enquote{{Bayesian Sparse Spiked Covariance Model with a Continuous Matrix Shrinkage Prior}.}
\newblock {\em Bayesian Analysis\/}, 17(4): 1193--1217.
\endbibitem

\bibitem[{Yoshida and West(2010)}]{Yoshida2010}
Yoshida, R. and West, M. (2010).
\newblock \enquote{Bayesian learning in sparse graphical factor models via annealed entropy.}
\newblock {\em Journal of Machine Learning Research\/}, 11: 1771--1798.
\endbibitem

\bibitem[{Yuchi et~al.(2023)Yuchi, Mak, and Xie}]{Yuchi2023}
Yuchi, H.~S., Mak, S., and Xie, Y. (2023).
\newblock \enquote{{Bayesian Uncertainty Quantification for Low-Rank Matrix Completion}.}
\newblock {\em Bayesian Analysis\/}, 18(2): 491--518.
\endbibitem

\bibitem[{Zhang et~al.(2025)Zhang, Cao, and Wang}]{zhang2023robustbayesianfunctionalprincipal}
Zhang, J., Cao, J., and Wang, L. (2025).
\newblock \enquote{Robust {B}ayesian functional principal component analysis.}
\newblock {\em Statistics and Computing\/}, 35(2): 46.
\endbibitem

\end{thebibliography}

\end{document}